\documentclass[review]{elsarticle}

\usepackage{lineno,hyperref}
\usepackage{graphicx,amsmath,amssymb,amsthm,bm}
\modulolinenumbers[5]

\journal{}

\DeclareMathOperator{\supp}{supp}

\newtheorem{theorem}{Theorem}
\newtheorem{proposition}{Proposition}

\begin{document}

\begin{frontmatter}

\title{Common knowledge equilibrium of Boolean securities in distributed information market}

\author{Masahiko Ueda\corref{mycorrespondingauthor}}
\address{Department of Systems Science, Graduate School of Informatics, Kyoto University, Kyoto 606-8501, Japan}
\ead{ueda.masahiko.5r@kyoto-u.ac.jp}

\begin{abstract}
We investigate common knowledge equilibrium of separable (or parity) and totally symmetric Boolean securities in distributed information market.
We theoretically show that clearing price converges to the true value when a common prior probability distribution of information of each player satisfies some conditions.
\end{abstract}

\begin{keyword}
Distributed information market; Common knowledge equilibrium
\end{keyword}

\end{frontmatter}


\section{Introduction}
\label{sec:introduction}
Common knowledge is that all agents know, that all agents know that all agents know, and so on ad infinitum \cite{FudTir1991,OsbRub1994}.
The concept of common knowledge was first introduced mathematically by Aumann in 1976 \cite{Aum1976}.
By using the concept, he proved that if all agents have the same prior distribution, and their posteriors for an event are common knowledge, then these posteriors are equal.
After his pioneering work, Geanakoplos and Polemarchakis proved that even though posteriors of agents are initially different, iterative announcement and revision process of posteriors leads to the state where posteriors of all agents are equal to each other in finite steps \cite{GeaPol1982}.
Furthermore, McKelvey and Page extended their results to the situation where not posteriors themselves but some statistics of posteriors are announced \cite{McKPag1986}.
They showed that if statistics of posteriors satisfies some condition, convergence to common knowledge occurs.
In addition, Nielsen et al. extended the results of Ref. \cite{McKPag1986} for conditional probability to conditional expectations \cite{NBGMP1990}.
Because, in many economic settings, it is more natural to suppose that only some aggregate of individual information (such as price) is announced instead of posteriors themselves, the results of Refs. \cite{McKPag1986,NBGMP1990} are useful in more realistic situation.

It has been considered that markets have power to compute the payoffs of securities \cite{PLGN2001}.
Feigenbaum et al. proposed a simple model of market where the payoff of some security is computed from information distributed in players through trades \cite{FFPS2005}.
They showed that the equilibrium of this model is described by the concept of common knowledge, and the necessary and sufficient condition for the market to correctly compute the Boolean payoff for all priors is that the payoff is described by a weighted threshold function.
This model was further investigated in Ref. \cite{CMC2006}, where effect of aggregate uncertainty was studied.
However, computational power of even such simple model has not been completely known.
For example, computational power of this model when convergence for all priors is not required has not been known.

In this paper, we investigate computational power of the distributed information market model \cite{FFPS2005} for two classes of Boolean securities, that is, separable (or parity) and totally symmetric, both of which are not necessarily the form of weighted threshold function.
We prove that the payoff of such securities can be correctly computed by market when prior is described by some form.

The paper is organized as follows.
In section \ref{sec:model}, we introduce a model of distributed information market.
In section \ref{sec:previous}, we review previous results for this model.
In section \ref{sec:separable}, we introduce the concept of separable (or parity) securities and prove that clearing price of these securities converges to the true value when a common prior probability distribution of information of each player is uniformly biased distribution.
In section \ref{sec:symmetric}, we introduce totally symmetric securities and prove that clearing price of these securities converges to the true value when a common prior probability distribution of information of each player is also totally symmetric and satisfies some condition.
Section \ref{sec:discussion} is devoted to concluding remarks.

\section{Model}
\label{sec:model}
We consider distributed information market model \cite{FFPS2005}.
A set of players is described as $\left\{ 1, \cdots, N \right\}$.
We assume that each player has one bit of information about the true state of the world (private information).
Private information of player $i$ is described as $\sigma_i \in \left\{ 1, -1 \right\}$ (not $\left\{ 0, 1 \right\}$, for convenience).
We also assume that payoff of traded security is a Boolean function.
For convenience, we write the payoff of security as $g(\bm{\sigma}) \in \left\{ 1, -1 \right\}$, where we have defined $\bm{\sigma}\equiv \left( \sigma_1, \cdots, \sigma_N \right)$.
(The original payoff is obtained as $(1+g)/2$.)
The functional form of $g(\bm{\sigma})$ is assumed to be common knowledge among all players.
Furthermore, we assume that all players have a common prior probability distribution $\mathcal{P}(\bm{\sigma})$ over the values of $\bm{\sigma}$.
At time $t\in \mathbb{Z}$, player $i$ bids $b_{i,t}$ according to the expectation of $g(\bm{\sigma})$ conditional on her private information and a set of $\bm{\sigma}$ consistent with previous clearing prices.
(This rule is obtained by the assumption that players are risk-neutral, myopic, and bid truthfully.)
Market-price formulation process is modeled by Shapley-Shubik market game \cite{ShaShu1977} with restriction.
Then, the clearing price at this round is 
\begin{eqnarray}
 c_{t+1} &=& \frac{1}{N} \sum_{i=1}^N b_{i,t}
\end{eqnarray}
and net money gain of player $i$ is $c_{t+1}-b_{i,t}$.
Probability distribution of $\bm{\sigma}$ for each player is updated via Bayes' rule.

Mathematically, the dynamics of this market when the true state of the world is $\bm{\sigma}$ is described as follows:
\begin{eqnarray}
 P_{i,t}\left( \hat{\bm{\sigma}}; \bm{\sigma} \right) &=& \frac{\delta_{\hat{\sigma}_i, \sigma_i} P^\mathrm{(ex)}_t\left( \hat{\bm{\sigma}}; \bm{\sigma} \right)}{\sum_{\hat{\bm{s}}} \delta_{\hat{s}_i, \sigma_i} P^\mathrm{(ex)}_t\left( \hat{\bm{s}}; \bm{\sigma} \right)} \\
 b_{i,t} \left( \bm{\sigma} \right) &=& \sum_{\hat{\bm{\sigma}}} g\left( \hat{\bm{\sigma}} \right) P_{i,t}\left( \hat{\bm{\sigma}}; \bm{\sigma} \right) \\
 c_t \left( \bm{\sigma} \right) &=& \frac{1}{N} \sum_{i=1}^N b_{i, t-1} \left( \bm{\sigma} \right) \\
 P^\mathrm{(ex)}_t\left( \hat{\bm{\sigma}}; \bm{\sigma} \right) &=& \frac{\mathbb{I}\left( c_t \left( \hat{\bm{\sigma}} \right) = c_t \left( \bm{\sigma} \right) \right) P^\mathrm{(ex)}_{t-1}\left( \hat{\bm{\sigma}}; \bm{\sigma} \right)}{\sum_{\hat{\bm{s}}} \mathbb{I}\left( c_t \left( \hat{\bm{s}} \right) = c_t \left( \bm{\sigma} \right) \right) P^\mathrm{(ex)}_{t-1}\left( \hat{\bm{s}}; \bm{\sigma} \right)}. \label{eq:dynamics_Pex}
\end{eqnarray}
with the initial condition
\begin{eqnarray}
 P^\mathrm{(ex)}_0\left( \hat{\bm{\sigma}}; \bm{\sigma} \right) &=& \mathcal{P} \left( \hat{\bm{\sigma}} \right).
\end{eqnarray}
Here we have introduced an indicator function $\mathbb{I}(\cdots)$ that returns $1$ when $\cdots$ holds and $0$ otherwise.
The function $P^\mathrm{(ex)}_t\left( \hat{\bm{\sigma}}; \bm{\sigma} \right)$ corresponds to the probability distribution of the state $\hat{\bm{\sigma}}$ at round $t$ for external observer when the true state of the world is $\bm{\sigma}$.
Similarly, the function $P_{i,t}\left( \hat{\bm{\sigma}}; \bm{\sigma} \right)$ corresponds to the probability distribution of the state $\hat{\bm{\sigma}}$ at round $t$ for player $i$.
$b_{i,t} \left( \bm{\sigma} \right)$ is a bid of player $i$ at round $t$.
$c_t \left( \bm{\sigma} \right)$ is the clearing price at round $t$.
We assume that the support of $\mathcal{P}\left( \hat{\bm{\sigma}} \right)$ contains $\bm{\sigma}$.
Generally, in order to calculate $P^\mathrm{(ex)}_t\left( \hat{\bm{\sigma}}; \bm{\sigma} \right)$, players and external observer need to compute $c_t \left( \bm{\sigma} \right)$ for all $\bm{\sigma} \in \left\{ 1, -1 \right\}^N$ and rule out $\bm{\sigma}$ that is inconsistent with the actual clearing price.

It should be noted that Eq. (\ref{eq:dynamics_Pex}) can be written as
\begin{eqnarray}
 P^\mathrm{(ex)}_t\left( \hat{\bm{\sigma}}; \bm{\sigma} \right) &=& \frac{\left\{ \prod_{t^\prime=1}^t \mathbb{I}\left( c_{t^\prime} \left( \hat{\bm{\sigma}} \right) = c_{t^\prime} \left( \bm{\sigma} \right) \right) \right\} \mathcal{P} \left( \hat{\bm{\sigma}} \right)}{\sum_{\hat{\bm{s}}} \left\{ \prod_{t^\prime=1}^t \mathbb{I}\left( c_{t^\prime} \left( \hat{\bm{s}} \right) = c_{t^\prime} \left( \bm{\sigma} \right) \right) \right\} \mathcal{P} \left( \hat{\bm{s}} \right)}.
 \label{eq:Pex_mod}
\end{eqnarray}
Equivalently, when we define a set
\begin{eqnarray}
 \mathcal{S}^t\left( \bm{\sigma} \right) &\equiv& \left\{ \left. \hat{\bm{\sigma}} \in \left\{ 1, -1 \right\}^N \right| c_t \left( \hat{\bm{\sigma}} \right) = c_t \left( \bm{\sigma} \right), \cdots, c_1 \left( \hat{\bm{\sigma}} \right) = c_1 \left( \bm{\sigma} \right) \right\},
\end{eqnarray}
Eq. (\ref{eq:Pex_mod}) can be rewritten as
\begin{eqnarray}
 P^\mathrm{(ex)}_t\left( \hat{\bm{\sigma}}; \bm{\sigma} \right) &=&  \left\{
\begin{array}{ll}
 \frac{\mathcal{P} \left( \hat{\bm{\sigma}} \right)}{\sum_{\hat{\bm{s}}\in \mathcal{S}^t\left( \bm{\sigma} \right)} \mathcal{P} \left( \hat{\bm{s}} \right)}  &\quad \left( \hat{\bm{\sigma}}\in \mathcal{S}^t\left( \bm{\sigma} \right) \right) \\
 0 &\quad \left( \hat{\bm{\sigma}}\notin \mathcal{S}^t\left( \bm{\sigma} \right) \right).
\end{array}
\right.
\end{eqnarray}
This implies $P^\mathrm{(ex)}_t\left( \cdot; \bm{\sigma} \right) = P^\mathrm{(ex)}_t\left( \cdot; \bm{\sigma}^\prime \right)$ for $\forall \bm{\sigma}^\prime \in \supp P^\mathrm{(ex)}_t\left( \cdot; \bm{\sigma} \right)=\mathcal{S}^t\left( \bm{\sigma} \right)$.

Furthermore, the constraint $c_t \left( \hat{\bm{\sigma}} \right) = c_t \left( \bm{\sigma} \right)$ for each $t$ in Eq. (\ref{eq:dynamics_Pex}) effectively gives a linear equation about $\hat{\bm{\sigma}}$ \cite{FFPS2005}.
In fact, for $\hat{\bm{\sigma}}\in \mathcal{S}^t\left( \bm{\sigma} \right)$ we obtain
\begin{eqnarray}
 b_{i,t} \left( \hat{\bm{\sigma}} \right) \mathbb{I}\left( \hat{\bm{\sigma}}\in \mathcal{S}^t\left( \bm{\sigma} \right) \right) &=& \frac{\sum_{\hat{\bm{s}}\in \mathcal{S}^t\left( \hat{\bm{\sigma}} \right)} g\left( \hat{\bm{s}} \right) \delta_{\hat{s}_i, \hat{\sigma}_i} \mathcal{P} \left( \hat{\bm{s}} \right)}{\sum_{\hat{\bm{s}}\in \mathcal{S}^t\left( \hat{\bm{\sigma}} \right)} \delta_{\hat{s}_i, \hat{\sigma}_i} \mathcal{P} \left( \hat{\bm{s}} \right)} \mathbb{I}\left( \hat{\bm{\sigma}}\in \mathcal{S}^t\left( \bm{\sigma} \right) \right) \nonumber \\
 &=& \frac{\sum_{\hat{\bm{s}}\in \mathcal{S}^t\left( \bm{\sigma} \right)} g\left( \hat{\bm{s}} \right) \delta_{\hat{s}_i, \hat{\sigma}_i} \mathcal{P} \left( \hat{\bm{s}} \right)}{\sum_{\hat{\bm{s}}\in \mathcal{S}^t\left( \bm{\sigma} \right)} \delta_{\hat{s}_i, \hat{\sigma}_i} \mathcal{P} \left( \hat{\bm{s}} \right)} \mathbb{I}\left( \hat{\bm{\sigma}}\in \mathcal{S}^t\left( \bm{\sigma} \right) \right) \nonumber \\
 &=& \left[ \beta_{i,t} \left( \mathcal{S}^t\left( \bm{\sigma} \right) \right) + \gamma_{i,t} \left( \mathcal{S}^t\left( \bm{\sigma} \right) \right) \hat{\sigma}_i \right] \mathbb{I}\left( \hat{\bm{\sigma}}\in \mathcal{S}^t\left( \bm{\sigma} \right) \right) \nonumber \\
 &&
\end{eqnarray}
with
\begin{eqnarray}
 \beta_{i,t} \left( \mathcal{S}^t\left( \bm{\sigma} \right) \right) &\equiv& \frac{1}{2} \left\{ \frac{\sum_{\hat{\bm{s}}\in \mathcal{S}^t\left( \bm{\sigma} \right)} g\left( \hat{\bm{s}} \right) \mathcal{P} \left( \hat{\bm{s}} \right) + \sum_{\hat{\bm{s}}\in \mathcal{S}^t\left( \bm{\sigma} \right)} g\left( \hat{\bm{s}} \right) \hat{s}_i \mathcal{P} \left( \hat{\bm{s}} \right)}{\sum_{\hat{\bm{s}}\in \mathcal{S}^t\left( \bm{\sigma} \right)} \mathcal{P} \left( \hat{\bm{s}} \right) + \sum_{\hat{\bm{s}}\in \mathcal{S}^t\left( \bm{\sigma} \right)} \hat{s}_i \mathcal{P} \left( \hat{\bm{s}} \right)} \right. \nonumber \\
 && \qquad \left. + \frac{\sum_{\hat{\bm{s}}\in \mathcal{S}^t\left( \bm{\sigma} \right)} g\left( \hat{\bm{s}} \right) \mathcal{P} \left( \hat{\bm{s}} \right) - \sum_{\hat{\bm{s}}\in \mathcal{S}^t\left( \bm{\sigma} \right)} g\left( \hat{\bm{s}} \right) \hat{s}_i \mathcal{P} \left( \hat{\bm{s}} \right)}{\sum_{\hat{\bm{s}}\in \mathcal{S}^t\left( \bm{\sigma} \right)} \mathcal{P} \left( \hat{\bm{s}} \right) - \sum_{\hat{\bm{s}}\in \mathcal{S}^t\left( \bm{\sigma} \right)} \hat{s}_i \mathcal{P} \left( \hat{\bm{s}} \right)} \right\} \label{eq:def_beta} \\
 \gamma_{i,t} \left( \mathcal{S}^t\left( \bm{\sigma} \right) \right) &\equiv& \frac{1}{2} \left\{ \frac{\sum_{\hat{\bm{s}}\in \mathcal{S}^t\left( \bm{\sigma} \right)} g\left( \hat{\bm{s}} \right) \mathcal{P} \left( \hat{\bm{s}} \right) + \sum_{\hat{\bm{s}}\in \mathcal{S}^t\left( \bm{\sigma} \right)} g\left( \hat{\bm{s}} \right) \hat{s}_i \mathcal{P} \left( \hat{\bm{s}} \right)}{\sum_{\hat{\bm{s}}\in \mathcal{S}^t\left( \bm{\sigma} \right)} \mathcal{P} \left( \hat{\bm{s}} \right) + \sum_{\hat{\bm{s}}\in \mathcal{S}^t\left( \bm{\sigma} \right)} \hat{s}_i \mathcal{P} \left( \hat{\bm{s}} \right)} \right. \nonumber \\
 && \qquad \left. - \frac{\sum_{\hat{\bm{s}}\in \mathcal{S}^t\left( \bm{\sigma} \right)} g\left( \hat{\bm{s}} \right) \mathcal{P} \left( \hat{\bm{s}} \right) - \sum_{\hat{\bm{s}}\in \mathcal{S}^t\left( \bm{\sigma} \right)} g\left( \hat{\bm{s}} \right) \hat{s}_i \mathcal{P} \left( \hat{\bm{s}} \right)}{\sum_{\hat{\bm{s}}\in \mathcal{S}^t\left( \bm{\sigma} \right)} \mathcal{P} \left( \hat{\bm{s}} \right) - \sum_{\hat{\bm{s}}\in \mathcal{S}^t\left( \bm{\sigma} \right)} \hat{s}_i \mathcal{P} \left( \hat{\bm{s}} \right)} \right\}. \label{eq:def_gamma}
\end{eqnarray}
Then, the constraint $c_{t+1} \left( \hat{\bm{\sigma}} \right) = c_{t+1} \left( \bm{\sigma} \right)$ for $\hat{\bm{\sigma}}\in \mathcal{S}^t\left( \bm{\sigma} \right)$ is equal to
\begin{eqnarray}
 \frac{1}{N} \sum_{i=1}^N \left[ \beta_{i,t} \left( \mathcal{S}^t\left( \bm{\sigma} \right) \right) + \gamma_{i,t} \left( \mathcal{S}^t\left( \bm{\sigma} \right) \right) \hat{\sigma}_i \right] &=& \frac{1}{N} \sum_{i=1}^N \left[ \beta_{i,t} \left( \mathcal{S}^t\left( \bm{\sigma} \right) \right) + \gamma_{i,t} \left( \mathcal{S}^t\left( \bm{\sigma} \right) \right) \sigma_i \right], \nonumber \\
 &&
\end{eqnarray}
which is linear with respect to $\left\{ \hat{\sigma}_i \right\}$.

\section{Previous studies}
\label{sec:previous}
For such market with information aggregation, properties of equilibrium ($t \rightarrow \infty$) have been investigated.
The next theorem is application of general theorem about common knowledge in Refs. \cite{McKPag1986,NBGMP1990} to the distributed information market model.
\begin{theorem}[Nielsen et al. \cite{NBGMP1990}]
\label{prop:CKE}
Suppose that the true state of the world is $\bm{\sigma}$.
At equilibrium, 
\begin{eqnarray}
 \sum_{\hat{\bm{\sigma}}} g\left( \hat{\bm{\sigma}} \right) P_{i,\infty}\left( \hat{\bm{\sigma}}; \bm{\sigma} \right) &=& \sum_{\hat{\bm{\sigma}}} g\left( \hat{\bm{\sigma}} \right) P^\mathrm{(ex)}_\infty\left( \hat{\bm{\sigma}}; \bm{\sigma} \right) = c_\infty \left( \bm{\sigma} \right)
\end{eqnarray}
holds for all $i$ when common prior is consistent (that is, the support of $\mathcal{P}\left( \hat{\bm{\sigma}} \right)$ contains $\bm{\sigma}$).
Furthermore, the convergence occurs in finite steps.
\end{theorem}
We call this equilibrium state as common knowledge equilibrium.

The next theorem gives the necessary and sufficient condition for $c_t \left( \bm{\sigma} \right)$ to converge to the true value $g\left( \bm{\sigma} \right)$ in $t \rightarrow \infty$ for arbitrary (consistent) prior distribution.
We denote step function as $\theta(\cdots)$.
\begin{theorem}[Feigenbaum et al. \cite{FFPS2005}]
\label{prop:CKE_threshold}
The necessary and sufficient condition for $c_t \left( \bm{\sigma} \right)$ to converge to the true value $g\left( \bm{\sigma} \right)$ in $t \rightarrow \infty$ for arbitrary (consistent) prior $\mathcal{P}\left( \hat{\bm{\sigma}} \right)$ is that $g\left( \bm{\sigma} \right)$ is written as a weighted threshold function
\begin{eqnarray}
 g\left( \bm{\sigma} \right) &=& 2\theta \left( \sum_{i=1}^N w_i \sigma_i - 1 \right) -1
\end{eqnarray}
with some real constants $w_1, \cdots, w_N$.
Furthermore, the convergence occurs after at most $N$ rounds.
\end{theorem}
For example, when $N=2$ and $g\left( \bm{\sigma} \right)=\sigma_1\sigma_2$ (which corresponds to XOR function), $c_t \left( \bm{\sigma} \right)$ does not converge to the true value $g\left( \bm{\sigma} \right)$ for uniform prior distribution $\mathcal{P}\left( \hat{\bm{\sigma}} \right)=1/4$.

\section{Separable security}
\label{sec:separable}
Because of Theorem \ref{prop:CKE_threshold}, clearing price of a security which cannot be written as a weighted threshold function does not necessarily converge to the true value $g\left( \bm{\sigma} \right)$.
In this paper, as one class of Boolean securities, we consider separable securities, which are of the form
\begin{eqnarray}
 g\left( \bm{\sigma} \right) &=& \prod_{i=1}^N g_i \left( \sigma_i \right)
\end{eqnarray}
and each $g_i$ takes the value in $\left\{ 1, -1 \right\}$.
The general form of $g_i$ is
\begin{eqnarray}
 g_i\left( \sigma_i \right) &=& (-1)^{a_i+r_i\frac{1-\sigma_i}{2}} \nonumber \\
 &=& (-1)^{a_i} \sigma_i^{r_i}
\end{eqnarray}
where $a_i\in \left\{ 0, 1 \right\}$ and $r_i\in \left\{ 0, 1 \right\}$.
Therefore, separable securities can also be called parity securities.
The number of separable securities is $2^{N+1}$, and it is much smaller than the number of all possible securities $2^{2^N}$.
Securities of this form contain XOR security, and therefore they are not necessarily written as a weighted threshold function.

It should be noted that when prior $\mathcal{P}\left( \hat{\bm{\sigma}} \right)$ depends on the true state of the world $\bm{\sigma}$ such as $\mathcal{P}\left( \hat{\bm{\sigma}} \right)=\delta_{\hat{\bm{\sigma}}, \bm{\sigma}}$, convergence to the true value $g\left( \bm{\sigma} \right)$ trivially occurs.
Therefore, we investigate only priors which assign non-zero probability for $\forall \hat{\bm{\sigma}} \in \left\{ 1, -1 \right\}^N$ and do not depend on $\bm{\sigma}$.

The next proposition is the first main result of this paper.
\begin{proposition}
\label{prop:CKE_separable}
Suppose that a security is separable and a common prior probability distribution is of the form (uniformly biased distribution)
\begin{eqnarray}
 \mathcal{P}\left( \hat{\bm{\sigma}} \right) &=& \prod_{i=1}^N \frac{e^{h\hat{\sigma}_i}}{2\cosh\left( h \right)}.
\end{eqnarray}
Then, for $h\neq 0$, $c_t \left( \bm{\sigma} \right)$ converges to the true value $g\left( \bm{\sigma} \right)$ in $t \rightarrow \infty$.
\end{proposition}

\begin{proof}
We explicitly calculate the time evolution.
First, bids of players at $t=0$ are
\begin{eqnarray}
 b_{i,0}\left( \bm{\sigma} \right) &=& \frac{\sum_{\hat{\bm{\sigma}}} \left\{ \prod_{j=1}^N (-1)^{a_j} \sigma_j^{r_j} \right\} \delta_{\hat{\sigma}_i, \sigma_i} \mathcal{P}\left( \hat{\bm{\sigma}} \right)}{\sum_{\hat{\bm{\sigma}}} \delta_{\hat{\sigma}_i, \sigma_i} \mathcal{P}\left( \hat{\bm{\sigma}} \right)} \nonumber \\
 &=& (-1)^{\sum_{j=1}^N a_j} \sigma_i^{r_i} \prod_{j\neq i} \tanh^{r_j} \left( h \right).
\end{eqnarray}
Then, the clearing price at the first round is
\begin{eqnarray}
 c_1 \left( \bm{\sigma} \right) &=& \frac{1}{N} (-1)^{\sum_{j=1}^N a_j} \sum_{i=1}^N \sigma_i^{r_i} \prod_{j\neq i} \tanh^{r_j} \left( h \right).
\end{eqnarray}
Next, posterior probability distribution for external observer at $t=1$ is
\begin{eqnarray}
 P^\mathrm{(ex)}_1\left( \hat{\bm{\sigma}}; \bm{\sigma} \right) &=&  \frac{\mathbb{I}\left( \sum_{i=1}^N \hat{\sigma}_i^{r_i} \prod_{j\neq i} \tanh^{r_j} \left( h \right) = \sum_{i=1}^N \sigma_i^{r_i} \prod_{j\neq i} \tanh^{r_j} \left( h \right) \right) \prod_{j=1}^N e^{h\hat{\sigma}_j}}{\sum_{\hat{\bm{s}}} \mathbb{I}\left( \sum_{i=1}^N \hat{s}_i^{r_i} \prod_{j\neq i} \tanh^{r_j} \left( h \right) = \sum_{i=1}^N \sigma_i^{r_i} \prod_{j\neq i} \tanh^{r_j} \left( h \right) \right) \prod_{j=1}^N e^{h\hat{s}_j}} \nonumber \\
 &=& \frac{\mathbb{I}\left( \sum_{k:r_k=1} \hat{\sigma}_k = \sum_{k:r_k=1} \sigma_k \right) \prod_{j=1}^N e^{h\hat{\sigma}_j}}{\sum_{\hat{\bm{s}}} \mathbb{I}\left( \sum_{k:r_k=1} \hat{s}_k = \sum_{k:r_k=1} \sigma_k \right) \prod_{j=1}^N e^{h\hat{s}_j}}.
\end{eqnarray}
In order to obtain the second line, we have used $r_j\in \left\{ 0, 1 \right\}$ and the assumption $h\neq 0$.
Then, for $h\neq 0$, bids of players at $t=1$ are
\begin{eqnarray}
 b_{i,1}\left( \bm{\sigma} \right) &=& \frac{\sum_{\hat{\bm{\sigma}}} \left\{ \prod_{j=1}^N (-1)^{a_j} \hat{\sigma}_j^{r_j} \right\} \delta_{\hat{\sigma}_i, \sigma_i} P^\mathrm{(ex)}_1\left( \hat{\bm{\sigma}}; \bm{\sigma} \right)}{\sum_{\hat{\bm{\sigma}}} \delta_{\hat{\sigma}_i, \sigma_i} P^\mathrm{(ex)}_1\left( \hat{\bm{\sigma}}; \bm{\sigma} \right)} \nonumber \\
 &=& \frac{\sum_{\hat{\bm{\sigma}}} \left\{ \prod_{j=1}^N (-1)^{a_j} \hat{\sigma}_j^{r_j} \right\} \delta_{\hat{\sigma}_i, \sigma_i} \mathbb{I}\left( \sum_{k:r_k=1} \hat{\sigma}_k = \sum_{k:r_k=1} \sigma_k \right) \prod_{j=1}^N e^{h\hat{\sigma}_j}}{\sum_{\hat{\bm{\sigma}}} \delta_{\hat{\sigma}_i, \sigma_i} \mathbb{I}\left( \sum_{k:r_k=1} \hat{\sigma}_k = \sum_{k:r_k=1} \sigma_k \right) \prod_{j=1}^N e^{h\hat{\sigma}_j}} \nonumber \\
 &=& \frac{\left\{ \prod_{j=1}^N (-1)^{a_j} \right\} \sigma_i^{r_i} \sum_{\left\{ \hat{\sigma}_j \right\}_{j\neq i}}^{r_j=1} \left\{ \prod_{j:r_j=1}^{j\neq i} \hat{\sigma}_j \right\} \mathbb{I}\left( \sum_{k:r_k=1}^{k\neq i} \hat{\sigma}_k = \sum_{k:r_k=1}^{k\neq i} \sigma_k \right) \prod_{j:r_j=1}^{j\neq i} e^{h\hat{\sigma}_j}}{\sum_{\left\{ \hat{\sigma}_j \right\}_{j\neq i}}^{r_j=1} \mathbb{I}\left( \sum_{k:r_k=1}^{k\neq i} \hat{\sigma}_k = \sum_{k:r_k=1}^{k\neq i} \sigma_k \right) \prod_{j:r_j=1}^{j\neq i} e^{h\hat{\sigma}_j}} \nonumber \\
 &=& \frac{\left\{ \prod_{j=1}^N (-1)^{a_j} \right\} \sigma_i^{r_i} \sum_{\left\{ \hat{\sigma}_j \right\}_{j\neq i}}^{r_j=1} \left\{ \prod_{j:r_j=1}^{j\neq i} (-1)^{\frac{1-\hat{\sigma}_j}{2}} \right\} \mathbb{I}\left( \sum_{k:r_k=1}^{k\neq i} \hat{\sigma}_k = \sum_{k:r_k=1}^{k\neq i} \sigma_k \right) \prod_{j:r_j=1}^{j\neq i} e^{h\hat{\sigma}_j}}{\sum_{\left\{ \hat{\sigma}_j \right\}_{j\neq i}}^{r_j=1} \mathbb{I}\left( \sum_{k:r_k=1}^{k\neq i} \hat{\sigma}_k = \sum_{k:r_k=1}^{k\neq i} \sigma_k \right) \prod_{j:r_j=1}^{j\neq i} e^{h\hat{\sigma}_j}} \nonumber \\
 &=& \prod_{j=1}^N (-1)^{a_j} \sigma_j^{r_j} \nonumber \\
 &=& g\left( \bm{\sigma} \right),
 \label{eq:separable_b1}
\end{eqnarray}
which implies that convergence to the true price occurs.
We note that in the third line, we have calculated the sum with respect to $\left\{ \hat{\sigma}_j \right\}$ with $r_j=0$ or $j=i$ in both denominator and numerator.
In fact, at $t=2$
\begin{eqnarray}
 c_2 \left( \bm{\sigma} \right) &=& g\left( \bm{\sigma} \right)
\end{eqnarray}
\begin{eqnarray}
 P^\mathrm{(ex)}_2\left( \hat{\bm{\sigma}}; \bm{\sigma} \right) &=& \frac{\mathbb{I}\left( g\left( \hat{\bm{\sigma}} \right) = g\left( \bm{\sigma} \right) \right) \mathbb{I}\left( c_1\left( \hat{\bm{\sigma}} \right) = c_1\left( \bm{\sigma} \right) \right) \mathcal{P}\left( \hat{\bm{\sigma}} \right)}{\sum_{\hat{\bm{s}}} \mathbb{I}\left( g\left( \hat{\bm{s}} \right) = g\left( \bm{\sigma} \right) \right) \mathbb{I}\left( c_1\left( \hat{\bm{s}} \right) = c_1\left( \bm{\sigma} \right) \right) \mathcal{P}\left( \hat{\bm{s}} \right)}
\end{eqnarray}
and
\begin{eqnarray}
 b_{i,2}\left( \bm{\sigma} \right) &=& \frac{\sum_{\hat{\bm{\sigma}}} g\left( \hat{\bm{\sigma}} \right) \delta_{\hat{\sigma}_i, \sigma_i} P^\mathrm{(ex)}_2\left( \hat{\bm{\sigma}}; \bm{\sigma} \right)}{\sum_{\hat{\bm{\sigma}}} \delta_{\hat{\sigma}_i, \sigma_i} P^\mathrm{(ex)}_2\left( \hat{\bm{\sigma}}; \bm{\sigma} \right)} \nonumber \\
 &=& g\left( \bm{\sigma} \right),
\end{eqnarray}
and convergence indeed occurs.
\end{proof}

This result suggests that when there is a common trend in private information of players, convergence to the true value occurs.
It should be noted that when $h=0$, $b_{i,0}\left( \bm{\sigma} \right)=0$ and convergence to the true price does not occur.

\section{Totally symmetric security}
\label{sec:symmetric}
Here, as another class of Boolean securities, we investigate totally symmetric securities.
When we define $\bm{\sigma}_\pi \equiv \left( \sigma_{\pi(1)}, \cdots, \sigma_{\pi(N)} \right)$ with a permutation $\pi$ on $\left\{ 1, \cdots, N \right\}$, a totally symmetric security is defined as a security with
\begin{eqnarray}
 g\left( \bm{\sigma}_\pi \right) &=& g\left( \bm{\sigma} \right)
\end{eqnarray}
for arbitrary permutation $\pi$.
The general form of totally symmetric securities is
\begin{eqnarray}
 g\left( \bm{\sigma} \right) &=& \sum_{k=0}^N A_k \mathbb{I}\left( \sum_{j=1}^N \sigma_j = -N+2k \right)
 \label{eq:form_symmetric}
\end{eqnarray}
with $A_k\in  \left\{ 1, -1 \right\}$.
The number of totally symmetric securities is $2^{N+1}$.
Totally symmetric securities are not necessarily written as a weighted threshold function.

The next proposition is the second main result of this paper.
\begin{proposition}
\label{prop:CKE_symmetric}
Suppose that a security is totally symmetric and a common prior probability distribution is also totally symmetric, that is,
\begin{eqnarray}
 \mathcal{P}\left( \hat{\bm{\sigma}}_\pi \right) &=& \mathcal{P}\left( \hat{\bm{\sigma}} \right)
\end{eqnarray}
for arbitrary permutation $\pi$.
Then, if the relation
\begin{eqnarray}
 \sum_{\hat{\bm{s}}} g\left( \hat{\bm{s}} \right) \hat{s}_i \mathcal{P} \left( \hat{\bm{s}} \right) &\neq& \left\{ \sum_{\hat{\bm{s}}} g\left( \hat{\bm{s}} \right) \mathcal{P} \left( \hat{\bm{s}} \right) \right\} \left\{ \sum_{\hat{\bm{s}}} \hat{s}_i \mathcal{P} \left( \hat{\bm{s}} \right) \right\} \label{eq:assumption_symmetric}
\end{eqnarray}
holds (the both-hand sides of which do not depend on $i$), $c_t \left( \bm{\sigma} \right)$ converges to the true value $g\left( \bm{\sigma} \right)$ in $t \rightarrow \infty$.
\end{proposition}

\begin{proof}
We explicitly calculate the time evolution.
First, bids of players at $t=0$ are
\begin{eqnarray}
 b_{i,0}\left( \bm{\sigma} \right) &=& \frac{\sum_{\hat{\bm{\sigma}}} g\left( \hat{\bm{\sigma}} \right) \delta_{\hat{\sigma}_i, \sigma_i} \mathcal{P}\left( \hat{\bm{\sigma}} \right)}{\sum_{\hat{\bm{\sigma}}} \delta_{\hat{\sigma}_i, \sigma_i} \mathcal{P}\left( \hat{\bm{\sigma}} \right)} \nonumber \\
 &=& \beta_0 + \gamma_0 \sigma_i
\end{eqnarray}
where we have defined
\begin{eqnarray}
 \beta_0 &\equiv& \frac{1}{2} \left\{ \frac{\sum_{\hat{\bm{s}}} g\left( \hat{\bm{s}} \right) \mathcal{P} \left( \hat{\bm{s}} \right) + \sum_{\hat{\bm{s}}} g\left( \hat{\bm{s}} \right) \hat{s}_i \mathcal{P} \left( \hat{\bm{s}} \right)}{\sum_{\hat{\bm{s}}} \mathcal{P} \left( \hat{\bm{s}} \right) + \sum_{\hat{\bm{s}}} \hat{s}_i \mathcal{P} \left( \hat{\bm{s}} \right)} + \frac{\sum_{\hat{\bm{s}}} g\left( \hat{\bm{s}} \right) \mathcal{P} \left( \hat{\bm{s}} \right) - \sum_{\hat{\bm{s}}} g\left( \hat{\bm{s}} \right) \hat{s}_i \mathcal{P} \left( \hat{\bm{s}} \right)}{\sum_{\hat{\bm{s}}} \mathcal{P} \left( \hat{\bm{s}} \right) - \sum_{\hat{\bm{s}}} \hat{s}_i \mathcal{P} \left( \hat{\bm{s}} \right)} \right\} \nonumber \\
 && \label{eq:beta0_sym} \\
 \gamma_0 &\equiv& \frac{1}{2} \left\{ \frac{\sum_{\hat{\bm{s}}} g\left( \hat{\bm{s}} \right) \mathcal{P} \left( \hat{\bm{s}} \right) + \sum_{\hat{\bm{s}}} g\left( \hat{\bm{s}} \right) \hat{s}_i \mathcal{P} \left( \hat{\bm{s}} \right)}{\sum_{\hat{\bm{s}}} \mathcal{P} \left( \hat{\bm{s}} \right) + \sum_{\hat{\bm{s}}} \hat{s}_i \mathcal{P} \left( \hat{\bm{s}} \right)} - \frac{\sum_{\hat{\bm{s}}} g\left( \hat{\bm{s}} \right) \mathcal{P} \left( \hat{\bm{s}} \right) - \sum_{\hat{\bm{s}}} g\left( \hat{\bm{s}} \right) \hat{s}_i \mathcal{P} \left( \hat{\bm{s}} \right)}{\sum_{\hat{\bm{s}}} \mathcal{P} \left( \hat{\bm{s}} \right) - \sum_{\hat{\bm{s}}} \hat{s}_i \mathcal{P} \left( \hat{\bm{s}} \right)} \right\} \nonumber \\
 && \label{eq:gamma0_sym}
\end{eqnarray}
according to Eqs. (\ref{eq:def_beta}) and (\ref{eq:def_gamma}).
It should be noted that the right-hand sides of Eqs. (\ref{eq:beta0_sym}) and (\ref{eq:gamma0_sym}) do not depend on $i$ because $g$ and $\mathcal{P}$ are totally symmetric.
Then, the clearing price at the first round is
\begin{eqnarray}
 c_1 \left( \bm{\sigma} \right) &=& \beta_0 + \gamma_0 \frac{1}{N} \sum_{i=1}^N \sigma_i.
\end{eqnarray}
Next, posterior probability distribution for external observer at $t=1$ is
\begin{eqnarray}
 P^\mathrm{(ex)}_1\left( \hat{\bm{\sigma}}; \bm{\sigma} \right) &=&  \frac{\mathbb{I}\left( \gamma_0 \sum_{i=1}^N \hat{\sigma}_i = \gamma_0 \sum_{i=1}^N \sigma_i \right) \mathcal{P}\left( \hat{\bm{\sigma}} \right)}{\sum_{\hat{\bm{s}}} \mathbb{I}\left( \gamma_0 \sum_{i=1}^N \hat{s}_i = \gamma_0 \sum_{i=1}^N \sigma_i \right) \mathcal{P}\left( \hat{\bm{s}} \right)}.
\end{eqnarray}
By using the assumption (\ref{eq:assumption_symmetric}), that is, $\gamma_0\neq 0$, we obtain
\begin{eqnarray}
 P^\mathrm{(ex)}_1\left( \hat{\bm{\sigma}}; \bm{\sigma} \right) &=&  \frac{\mathbb{I}\left( \sum_{i=1}^N \hat{\sigma}_i = \sum_{i=1}^N \sigma_i \right) \mathcal{P}\left( \hat{\bm{\sigma}} \right)}{\sum_{\hat{\bm{s}}} \mathbb{I}\left( \sum_{i=1}^N \hat{s}_i = \sum_{i=1}^N \sigma_i \right) \mathcal{P}\left( \hat{\bm{s}} \right)}.
\end{eqnarray}
Then, bids of players at $t=1$ are
\begin{eqnarray}
 b_{i,1}\left( \bm{\sigma} \right) &=& \frac{\sum_{\hat{\bm{\sigma}}} g\left( \hat{\bm{\sigma}} \right) \delta_{\hat{\sigma}_i, \sigma_i} P^\mathrm{(ex)}_1\left( \hat{\bm{\sigma}}; \bm{\sigma} \right)}{\sum_{\hat{\bm{\sigma}}} \delta_{\hat{\sigma}_i, \sigma_i} P^\mathrm{(ex)}_1\left( \hat{\bm{\sigma}}; \bm{\sigma} \right)} \nonumber \\
 &=& \frac{\sum_{\hat{\bm{\sigma}}} g\left( \hat{\bm{\sigma}} \right) \delta_{\hat{\sigma}_i, \sigma_i} \mathbb{I}\left( \sum_{j=1}^N \hat{\sigma}_j = \sum_{j=1}^N \sigma_j \right) \mathcal{P}\left( \hat{\bm{\sigma}} \right)}{\sum_{\hat{\bm{\sigma}}} \delta_{\hat{\sigma}_i, \sigma_i} \mathbb{I}\left( \sum_{j=1}^N \hat{\sigma}_j = \sum_{j=1}^N \sigma_j \right) \mathcal{P}\left( \hat{\bm{\sigma}} \right)} \nonumber \\
 &=& \frac{\sum_{\hat{\bm{\sigma}}} g\left( \bm{\sigma} \right) \delta_{\hat{\sigma}_i, \sigma_i} \mathbb{I}\left( \sum_{j=1}^N \hat{\sigma}_j = \sum_{j=1}^N \sigma_j \right) \mathcal{P}\left( \hat{\bm{\sigma}} \right)}{\sum_{\hat{\bm{\sigma}}} \delta_{\hat{\sigma}_i, \sigma_i} \mathbb{I}\left( \sum_{j=1}^N \hat{\sigma}_j = \sum_{j=1}^N \sigma_j \right) \mathcal{P}\left( \hat{\bm{\sigma}} \right)} \nonumber \\
 &=& g\left( \bm{\sigma} \right),
 \label{eq:symmetric_b1}
\end{eqnarray}
where we have used the fact that the value of a totally symmetric Boolean function is determined only by $\sum_{j=1}^N \hat{\sigma}_j$ (Eq. (\ref{eq:form_symmetric})).
This implies that convergence to the true price occurs (similarly to Proposition \ref{prop:CKE_separable}).
\end{proof}

We remark that convergence to the true price does not occur when $\gamma_0=0$.

\section{Discussion}
\label{sec:discussion}
In this paper, we have investigated equilibrium of iterative process in which the clearing price is publicly announced and players revise their bids according to the public information and their own private information, in distributed information market model with two classes of Boolean securities, that is, separable and totally symmetric.
As is well known, the equilibrium state of this model is described by the concept of common knowledge.
We have theoretically showed that the clearing price of separable securities converges to the true value when a common prior probability distribution of information of each player is uniformly biased distribution.
In contrast, when a common prior probability distribution is uniform distribution over $\left\{ 1, -1 \right\}^N$, convergence to the true value does not occur.
We have also theoretically showed that the clearing price of totally symmetric securities converges to the true value when a common prior probability distribution of information of each player satisfies some condition.

The convergence to the true price in separable and totally symmetric securities seems to come from the fact that structure of $g^{-1}(1)$ and $g^{-1}(-1)$ is simple.
As we can see in Eqs. (\ref{eq:separable_b1}) and (\ref{eq:symmetric_b1}), although convergence to the true price occurs, players cannot know the true state $\bm{\sigma}$, and convergence seems to come from degeneracy of the set of $\bm{\sigma}$ with the same clearing price in $g^{-1}(1)$ or $g^{-1}(-1)$.
For securities which are not separable or totally symmetric, situation will be more complicated.
Finding general priors for non-separable or non-totally-symmetric securities and elucidating the relation between geometry of $g^{-1}(1)$ and $g^{-1}(-1)$ and appropriate priors for convergence to the true price is an important future problem.

In this paper, we only considered noiseless situation, where each process of the time evolution is accurate and precise.
However, this assumption is not realistic.
If noise exists, convergence of the iterative process of this model to the true state may not occur, by convergence to wrong states.
There are various origins of noises in game-theoretic situations \cite{ANT2018}.
In our setting, significant noises may come from rounding error in calculation of clearing price, effect of irrational players, and incompleteness of information of players.
Related to the last point, authors of Ref. \cite{CMC2006} investigated the situation where the state of the world cannot be fully determined even if information of all players is pooled together.
They found that convergence property of the distributed information market model becomes worse in specific examples when such aggregate uncertainty exists.
Investigating whether previous results and our result can be extended to noisy situations or not is another important future problem.

Related to the above remark, considering learning process in Eq. (\ref{eq:dynamics_Pex}) would be interesting.
Because players and external observer need to compute $c_t \left( \bm{\sigma} \right)$ for all $\bm{\sigma} \in \left\{ 1, -1 \right\}^N$ and rule out $\bm{\sigma}$ that is inconsistent with the actual clearing price, it needs much computational costs.
It is more realistic that players gradually learn $\bm{\sigma}$ by calculation with low computational costs (that is, players are bounded-rational).
Research in this direction is needed.

Furthermore, we are also interested in common knowledge equilibrium of other model such as market scoring rule \cite{Han2003,Han2007,CDSRPHFG2010}.
Market scoring rule is another model of information market (or prediction market), and it is myopically incentive compatible.
In addition, when logarithmic market scoring rule is adopted, analysis in terms of information theory is possible \cite{Han2007,CDSRPHFG2010}.
We will perform information theoretical analysis of the common knowledge equilibrium in future.

\section*{Acknowledgement}
This study was supported by JSPS KAKENHI Grant Numbers JP19K21542 and JP20K19884.


\section*{References}

\bibliography{DIM}

\end{document}